\documentclass[copyright]{eptcs}
\providecommand{\event}{QPL 2015} 
\providecommand{\urlalt}[2]{\href{#1}{#2}} 
\providecommand{\doi}[1]{doi:\urlalt{http://dx.doi.org/#1}{#1}}

\title{Entropy, majorization and thermodynamics 
in general probabilistic theories}
\author{Howard Barnum
\institute{University of New Mexico \\ Albuquerque, NM}
\email{hnbarnum@aol.com}
\and
Jonathan Barrett \qquad\qquad 
\institute{Department of Computer Science, University of Oxford\\
Oxford, U.K.}
\email{\quad jonathan.barrett@cs.ox.ac.uk}
\and
Marius Krumm
\institute{Department of Theoretical Physics, University of Heidelberg, 
Heidelberg, Germany}
\email{\quad M.Krumm@stud.uni-heidelberg.de}
\and
Markus P.\ M\"uller
\institute{Department of Theoretical Physics, University of Heidelberg, 
Heidelberg, Germany}
\institute{Department of Applied Mathematics, University of Western Ontario, London ON, Canada}
\institute{Department of Philosophy, University of Western Ontario, London ON, Canada}
\institute{Perimeter Institute for Theoretical Physics, Waterloo, Canada}
\email{\quad markus@mpmueller.net}
}

\def\titlerunning{Entropy, majorization and thermodynamics in general probabilistic theories}
\def\authorrunning{H. Barnum, J. Barrett, M. Krumm, M. M\"uller}

\usepackage{amsfonts, amssymb}
\usepackage{amsthm}
\usepackage{amsmath}
\usepackage[usenames]{color}
\usepackage{epsfig}
\usepackage{xcolor}
\usepackage{bm}
\usepackage{bbm}
\usepackage{enumerate}

\newcommand{\face}{\mbox{face}}

\newcommand{\im}{\mbox{im}}

\newtheorem{theorem}{Theorem}[section]
\newtheorem{lemma}[theorem]{Lemma}
\newtheorem{proposition}[theorem]{Proposition}
\newtheorem{corollary}[theorem]{Corollary}

\newtheorem{definition}[theorem]{Definition}

\newtheorem{question}{Question}

\renewcommand{\hat}{\widehat}

\renewcommand{\v}{{\bf v}}
\newcommand{\w}{{\bf w}}

\newcommand{\x}{{\bf x}}

\newcommand{\p}{{\bf p}}

\newcommand{\tr}{\text{tr\,}}

\renewcommand{\phi}{\varphi}

\newcommand{\beq}{\begin{equation}}
\newcommand{\eeq}{\end{equation}}
\newcommand{\beqa}{\begin{eqnarray}}
\newcommand{\eeqa}{\end{eqnarray}}
\def\qed{{\hfill $\blacksquare$}}

\newcommand{\cP}{{\cal P}}

\def\R{{\mathbb{R}}}

\def\lambdad{\lambda^{\downarrow}}
\def\lambdau{\lambda^{\uparrow}}

\newcommand{\ket}[1]{| #1 \rangle}
\newcommand{\bra}[1]{\langle #1 |}
\newcommand{\proj}[1]{\ket{#1}\! \bra{#1}}

\def\join{\vee}

\def\union{\cup}
\def\intersect{\cap}

\def\face{{\rm Face}}

\def\conv{{\rm Conv~}}

\def\ker{{\rm ker~}}
\def\lin{{\mathrm{lin~}}}

\newcommand{\hbcommentoff}[1]{{}}
\newcommand{\tempout}[1]{{}}

\begin{document}
\maketitle

\begin{abstract}
In this note we lay some groundwork for the resource theory of
thermodynamics in general probabilistic theories (GPTs). We consider
theories satisfying a purely convex abstraction of the spectral
decomposition of density matrices: that every state has a
decomposition, with unique probabilities, into perfectly
distinguishable pure states. The spectral entropy, and analogues using
other Schur-concave functions, can be defined as the entropy of these
probabilities. We describe additional conditions under which the
outcome probabilities of a fine-grained measurement are majorized by
those for a spectral measurement, and therefore the ``spectral entropy"
is the measurement entropy (and therefore concave). These conditions
are (1) projectivity, which abstracts aspects of the L\"uders-von
Neumann projection postulate in quantum theory, in particular that
every face of the state space is the positive part of the image of a
certain kind of projection operator called a filter; and (2) symmetry
of transition probabilities. The conjunction of these, as shown
earlier by Araki, is equivalent to a strong geometric property of the
unnormalized state cone known as perfection: that there is an inner
product according to which every face of the cone, including the cone
itself, is self-dual. Using some assumptions about the thermodynamic
cost of certain processes that are partially motivated by our
postulates, especially projectivity, we extend von Neumann's argument
that the thermodynamic entropy of a quantum system is its spectral
entropy to generalized probabilistic systems satisfying spectrality.

\end{abstract}

Much progress has recently been made (see for example
\cite{SecondLaws, LJR}) in understanding the fine-grained
thermodynamics and statistical mechanics of microscopic quantum
physical systems, using the fundamental idea of thermodynamics as a
particular type of \emph{resource theory}.  A resource theory is a
theory that governs which state transitions, whether deterministic or
stochastic, are possible in a given theory, using specified means.
This depends on the kind of state transformations that are allowed in
the theory, and in particular on which subset of them are specified as
``thermodynamically allowed''.

In this note, we lay some groundwork for the resource theory of
thermodynamics in general probabilistic theories (GPTs).  We describe
simple, but fairly strong, postulates under which the structure of
systems in a theory gives rise to a natural notion of spectrum,
allowing definition of entropy-like quantities and relations such as
majorization analogous to those involved in fine-grained quantum and
classical thermodynamics.  These are candidates for governing which
transitions are thermodynamically possible under specialized
conditions (such as a uniform (``microcanonical'') bath), and for
figuring in definitions of free-energy-like quantities that constrain
state transitions under more general conditions. In further
work~\cite{InPreparation} we will investigate the extent to which they
do so, under reasonable assumptions about which transformations are
thermodynamically allowed.  Since the postulates are shared by quantum
theory, and the spectrum of quantum states is deeply involved in
determining which transitions are possible, we expect that these
postulates, supplemented with further ones (including a notion of
energy), allow the development of a thermodynamics and statistical
mechanics fairly similar to the quantum one.

Our main result, proven in Section \ref{sec: majorization}, is that
under certain assumptions, implied by but weaker than the conjunction
of Postulates 1 and 2 of \cite{BMU}, the outcome probabilities of a
finegrained measurement are majorized by those for a spectral
measurement, and therefore the ``spectral entropy'' is the measurement
entropy (and therefore concave).  It also allows other entropy-like
quantities, based on Schur-concave functions, to be defined.  Our
first assumption, described in Section \ref{sec: spectrality} is a
purely convex abstraction of the spectral decomposition of density
matrices: that every state has a decomposition, with unique
probabilities, into perfectly distinguishable pure (i.e. extremal)
states.  The spectral entropy (and analogues using other Schur-concave
functions) can be defined as the entropy of these probabilities.
Another assumption, \emph{projectivity} (Section \ref{subsec: projectivity}), 
abstracts aspects of the projection postulate in quantum theory; 
together with \emph{symmetry of transition probabilities} it ensures 
the desirable behavior of the spectral entropic quantities that follows
from our main result.  

In Sections 5 and 6 we note that projectivity on its own implies a
spectral expansion for \emph{observables} (our additional spectrality
assumption is for \emph{states}), and also note the equivalence of the
premises of our theorem on spectra to a strong kind of self-duality,
known as \emph{perfection}, of the state space.  

Section 7 contains another main result of this work.  Using
spectrality, and some assumptions about the thermodynamic cost of
certain processes that are partially motivated by our other
postulates, especially projectivity, we generalize von Neumann's
argument that the thermodynamic entropy of a quantum system is its
spectral entropy, to generalized probabilistic systems satisfying
spectrality.  We then consider the prospect of embedding this result
in a broader thermodynamics of systems satisfying relevant properties
including the ones used in the present work, as well as others.  Among
the other useful properties, \emph{Energy Observability}, which was
used in \cite{BMU} to narrow down the class of Jordan algebraic
theories to standard complex quantum theory, can provide a
well-behaved notion of energy to play a role in a fuller thermodynamic
theory, and an ample automorphism group of the normalized space,
acting transitively on the extremal points, or even strongly, on the
sets of mutually distinguishable pure states (Strong Symmetry
(\cite{BMU}), may enable reversible adiabatic processes that can be
crucial to thermodynamic protocols.

While our postulates are strong and satisfied by quantum theory, it is
far from clear that, even supplemented by energy observability, they
constrain us to quantum theories: in \cite{BMU} the strong property of
no higher-order interference was used, along with the properties of
Weak Spectrality, Strong Symmetry, and Energy Observability, to obtain
complex quantum theory as the unique solution.  While it is possible that latter three
properties alone imply quantum theory, this would be a highly
nontrivial result and we consider it at least as likely that they do
not.

In the special case of assuming Postulates 1 and 2 of~\cite{BMU}, a
proof of our main theorem (Theorem~\ref{theorem: majorization}) has
appeared in one of the authors' Master thesis~\cite{Krumm}, where
several further results have been obtained.  We will elaborate on
this, and in particular on the physics as detailed in von Neumann's
thought experiment (cf.\ Section~\ref{SecvonNeumann}),
elsewhere~\cite{InPreparation}. Note also the very closely related,
but independent work of Chiribella and
Scandolo~\cite{Giulio1,Giulio2}. The main difference to our work is
that (in most cases) they assume a ``purification postulate'' (among
other postulates), and thus rely on a different set of axioms than we
do.  General-probabilistic thermodynamics has also been considered
in~\cite{MDV,MOD}, where entanglement entropy and its role in the
black-hole information problem has been analyzed. We hope that these
different ways of approaching generalized thermodynamics will help to
identify the main features of a probabilistic theory which are
necessary for consistent thermodynamics, and thus lead to a different,
possibly more physical, understanding of the structure of quantum
theory.

\section{Systems}

In this section, we recall the general notion of system that we 
will use as an abstract model of potential physical systems in a theory, 
and define several properties of a system that in following sections
will be related to the existence of a spectrum.    
We make a standing assumption of finite dimension throughout the paper
except where it is explicitly suspended (notably in Appendix
\ref{appendix: spectral expansion}).

A system will be a triple consisting of a finite-dimensional regular
cone $A_+$ in a real vector space $A$, a distinguished regular cone
$A^\sharp_+\subseteq A^*_+$ ($A^*_+ \subset A^*$ being the cone dual
to $A_+$), and a distinguished element $u$ in the interior of
$A^\sharp_+$.  A (convex) \emph{cone} in a finite-dimensional real
vector space is a subset closed under addition and nonnegative scalar
multiplication; it is \emph{regular} if it linearly generates $A$,
contains no nontrivial subspace of $A$, and is topologically closed.
Usually we will refer to such a system by the name of its ambient
vector space, e.g. $A$.  The normalized states are the elements $x \in
A_+$ for which $u(x)=1$; the set $\Omega$ of such states is compact
and convex, and forms a base for the cone.  Measurement outcomes,
called \emph{effects}, are linear functionals $e \in A_+^\sharp$
taking values in $[0,1]$ when evaluated on normalized states; a
measurement is a (finite, for present purposes) set of effects that
add up to $u$.  Below, we will assume that $A^\sharp_+ = A^*_+$,
although we are investigating whether this can be derived from our
other assumptions.  Allowed dynamical processes on states will usually
be taken to be \emph{positive maps}: linear maps $T$ such that $TA_+
\subset A_+$.  Such a map is an \emph{order-automorphism} if $TA_+ =
A_+$, and \emph{reversible} if $T\Omega = \Omega$.  An
\emph{order-isomorphism} $T: A \rightarrow B$ between ordered vector
spaces is an isomorphism of vector spaces with $TA_+ = B_+$.

An \emph{extremal ray} of a cone $A_+$ is a ray $\rho = \R_+ x$, for
some nonzero $x \in A_+$, such that no $y \in \rho$ is a positive
linear combination of distinct nonzero elements of $A_+$ not in $\rho$.
Equivalently, it is the set of nonnegative multiples of an extremal
state (also called pure state) of $\Omega$.  (Extremal points in a
convex set are those that cannot be written as nontrivial convex
combinations of elements of the set.) A cone is \emph{reducible} if $A
= A_1 \oplus A_2$, a nontrivial vector space direct sum, and all
extremal rays of $A_+$ are contained either in $A_1$ or $A_2$, and
irreducible if it is not reducible.  Information about which of the
summands $A_i$ of a reducible cone a state lies in can be considered
essentially classical; $A_i$ are like ``superselection sectors''.

\section{Spectrality} 
\label{sec: spectrality}

We say a set $\{ \omega_i \}$ of states is 
\emph{perfectly distinguishable}
if there is a measurement $\{e_i \}$ such that $e_i(\omega_j) = \delta_{ij}$.

{\bf Axiom WS:} (``Weak Spectrality'')
Every state $\omega$ has a convex 
decomposition $\omega = \sum_i p_i \omega_i$ 
into perfectly distinguishable pure states.

{\bf Axiom S:} (``Spectrality'' (or ``Unique Spectrality'')).  Every
state has a decomposition $\sum_i p_i \omega_i$ into perfectly
distinguishable pure states.  If it has more than one such
decomposition, the two use the same probabilities.  In other words, if
$\omega = \sum_{i=1}^N p_i \omega_i = \sum_{i= 1}^M q_i \rho_i$, where
both $\omega_i$ and $\rho_i$ are sets of perfectly distinguishable
pure states, $p_i, q_i > 0$, $\sum_i p_i = \sum_i q_i = 1$, then $M=N$
and there is a permutation $\sigma$ of $\{1,...,N\}$ such that for
each $i$, $p_{\sigma(i)} = q_i$.

\emph{A priori} Axiom {\bf S} is stronger than Axiom {\bf WS}.  Later
in this note we will give an example of a weakly spectral, but not 
uniquely spectral, system.

Note that {\bf WS} is Postulate 1 of \cite{BMU}.  Postulate 2 of \cite{BMU} is
Strong Symmetry ({\bf SS}): that every set of mutually distinguishable pure
states can be taken to any other such set of the same size by a
reversible transformation (affine automorphism of $\Omega$).    
{\bf WS} and Strong Symmetry together imply Axiom {\bf S} and
Axiom {\bf P} (that is, ``projectivity'' as defined below).  The converse is probably not true.  Indeed, Postulates 1
and 2 of \cite{BMU} imply the very strong property of
\emph{perfection}, which, as we note in Section 
\ref{sec: projective stp perfect} is equivalent to Axioms 
{\bf S}, {\bf P}, and Symmetry of Transition Probabilities.  

There are various ways to \emph{derive} weak spectrality, as for 
example in \cite{Giulio2}, or \cite{Wilce10, Wilce09}.  

\section{Projective and perfect systems} 

\subsection{Projectivity}
\label{subsec: projectivity}

We call a finite-dimensional system \emph{projective} if each face of
$\Omega$ is the positive normalized part of the image of a \emph{filter}.
Filters are defined in \cite{BMU} to be normalized, bicomplemented,
positive linear projections $A \rightarrow A$.  This is equivalent to
being the dual of a \emph{compression}, where the latter is as defined in
\cite{ASBook}.  Normalization just means that they are contractive in
(do not increase) the base norm, which for $x \in A_+$ is just $u(x)$.
Recall that $P$ positive means 
$P A_+ \subseteq A_+$ and 
$P$ a projection means $P^2=P$.  We write $\im_+ P$ for $\im P \intersect A_+$, 
and $\ker_+ P$ for $\ker P \intersect A_+$. 
Then complemented means there is another positive projection $P'$ with 
$\im_+ P = \ker_+ P'$ and $\ker_+ P = \im_+ P'$, 
and bicomplemented means complemented with
complemented adjoint.  It can be shown that 
filters are neutral: if $u(x) = u(Px)$ (``$x$
passes the filter with certainty'') then $Px = x$ (``$x$ passes the
filter undisturbed'').  The complement, $P'$, is unique. 
The projections $P: X \mapsto Q X Q$ of
quantum theory, where $Q$ is an orthogonal projector onto a subspace of the
Hilbert space, are examples of filters.  The existence of filters 
might be important for informational protocols such as data compression,
or for thermodynamic protocols or the machinations of Maxwell demons.
In finite dimension, a system is projective in this sense
if and only if it satisfies the standing hypothesis of 
\cite{ASBook}, Ch. 8.  

A system is said to satisfy Axiom {\bf P} (``Projectivity'') if
it is projective.  The effects $u \circ P$, for filters $P$, are 
called \emph{projective units}.

\begin{proposition}[\cite{ASBook}, Theorem 8.1]
\label{prop: orthomodular}
For a projective state space, the lattice of faces is complete and
orthomodular.  The filters and the projective units, being in
one-to-one correspondence with faces, can be equipped with an
orthocomplemented lattice structure isomorphic to that of the faces.
For orthogonal faces $F$ and $G$, $u_{F \join G} = u_F + u_G$.
\end{proposition}

The relevant orthocomplementation is the map $P \rightarrow P'$ described
in the definition of filter above; by Proposition \ref{prop: orthomodular} 
it transfers to the lattices of faces and of projective units.

\subsection{Self-duality and perfection}
\label{subsec: perfection}

A regular cone $A_+$ is said to be \emph{self-dual} if there exists an
inner product $(.,.)$ on $A$ such that $A_+ = \{ y \in A: \forall x \in A_+, ~
( y , x ) \ge 0 \}$.  
(We sometimes refer to the RHS of this 
expression, even when $A_+$ is not self-dual, 
as the \emph{internal dual cone} of $A_+$ relative to the inner 
product, since it is affinely isomorphic to the dual cone.)  
This is
equivalent to the existence of an order isomorphism $\phi: A^*
\rightarrow A$ such that bilinear form $\langle . , \phi(.) \rangle$
is an inner product on $A$.  It is stronger than just order-isomorphism 
between $A$ and $A^*$, since we may have $\phi(A^*_+)=A_+$ without the nondegenerate
bilinear form $\langle . , \phi( . ) \rangle$ being positive definite (for
example, if $A_+$ is the cone with square base).     

\begin{definition}
A cone $A_+ \subset A$ 
is called \emph{perfect} if we can introduce a fixed inner product on $A$ 
such that each face of the cone (including the cone itself) 
is self-dual with respect to the restriction of that inner 
product to the span of the face.
\end{definition}

For such cones, the orthogonal (in the self-dualizing inner product)
projection $P_F$ onto the span of a face is a positive linear map
\cite{IochumThesis, IochumBook}.
It is clearly bicomplemented.  If the system has a distinguished order 
unit, with respect to which $P_F$ is normalized, then $P_F$ is a filter.

\begin{definition}
A \emph{perfect system} is one whose state space $A_+$ is perfect
and for which each of the orthogonal projections $P_F$ onto $\lin F$ is
normalized.
\end{definition}

It follows from this definition that the projections $P_F$ of a perfect
system are filters, hence a perfect system is projective.

\begin{question}
For a perfect cone, is there always a choice of order unit that makes
it projective?
\end{question}

\begin{question}
Are there perfect cones that can be equipped with order units that 
make them projective in inequivalent ways?
\end{question}

One may investigate these questions by looking at an analogue of
\emph{tracial states} (Def. 8.1 and the remark following it in
\cite{ASBook}).  In an appropriate setting (which includes systems
with spectral duality (\cite{ASBook}, Def. 8.42) 
in general, and is equivalent to projective
systems in the finite-dimensional case) a tracial state $\omega$ is one that is
central, i.e. such that $(P+P')\omega = \omega$ for all filters $P$.
Equivalently it is the intersection of $\conv(F \union F')$ for all
projective faces $F$.  The conditions $(P + P' ) \omega = \omega$ are
linear, so this defines a subspace of the state space; in a self-dual
cone, it also gives a subspace of the observables, and it is natural
to ask whether the order unit lies in that subspace, and whether,
indeed, in an irreducible self-dual projective cone the tracial states
are just the one-dimensional space generated by the order unit. 
If so, that suggests that we consider an analogue of the notion of the
linear space generated by tracial elements, for perfect cones: the
linear space spanned by states (or effects) such that $(P_F +
P_{F'})\omega = \omega$ for every orthogonal projection $P_F$ onto the
span of a face $F$.  We call the nonnegative elements of such a linear
space in a perfect cone \emph{orthotracial}.

\begin{proposition}
A system with $A_+$ a perfect cone and an orthotracial element $e$ in
the interior of $A_+$ taken as the order unit is the same thing as a
perfect projective system.
\end{proposition}

\begin{proof}
Let $e$ be orthotracial and in the interior of 
$A_+$.
Orthotraciality of $e$
says $P_F e + P_{F'} e = e$, i.e. $e - P_{F} e = P_{F'}e$.  
Since in any perfect cone the orthogonal
projections onto spans of faces are positive, $P_{F'} e \ge 0$; 
hence $e - P_Fe \ge 0$, i.e. $P_F e \le e$, i.e.  $P_F$ is normalized.
Conversely, in a perfect projective system the order unit is
orthotracial (as well as tracial).  We have already pointed out that
the orthogonal projectors $P$ onto spans of faces are compressions/filters
in this context; from $p + p' = u$, and $p := Pu, p' :=P'u$, we have
$(P + P' )u = u$ for all filters $P$.
\qed
\end{proof}

\begin{question}
Is an orthotracial state automatically in the interior of $A_+$?  
\end{question}

\section{Measurements, measurement entropy, and majorization}
\label{sec: majorization}

\begin{definition}
Let Axiom {\bf WS} hold.  A \emph{spectral measurement} on state $\omega$
is a measurement that distinguishes the pure states $\omega_i$ appearing in a 
convex decomposition of $\omega$.  
\end{definition}

Consider a system satisfying Axiom {\bf S} whose normalized state space
$\Omega$ has maximal number of perfectly distinguishable pure states
$n$.  Call a function $f: \R^n \rightarrow \R$ \emph{symmetric} if
$f(\x) = f(\sigma(\x))$ for any permutation $\sigma$.  For any
symmetric function $f: \R^n \rightarrow \R$, we define another
function, $f: \Omega \rightarrow \R$, by $f(\omega) = f(\p)$ where
$\p \in \R^n$ are the probabilities in a decomposition of 
$f$ into perfectly distinguishable pure states (extended by adding zeros if
the decomposition uses fewer than $n$ states).  By symmetry and unique
spectrality this is independent of the choice of decomposition, so our
claim that it defines a function is legitimate.  Define the functions
$\lambdad: \Omega \rightarrow \R^n$ and $\lambdau: \Omega \rightarrow
\R^n$ to take a state and return the decreasingly-ordered and
increasingly-ordered decomposition probabilities, respectively, of
$\omega$.  Then $f(\omega) = f(\lambdad(\omega)) =
f(\lambdau(\omega))$.

\begin{definition}
For $x, y \in \R^n$, $x \prec y$, ``$x$ is majorized by $y$'', means that
$\sum_{i=1}^k x^\downarrow_i \le \sum_{i=1}^k y^\downarrow_i$
for $k = 1,...,n-1$, and $\sum_{i=1}^n x^\downarrow_i = \sum_{i=1}^n y^\downarrow_i$.
\end{definition}

If the first condition holds, and the second holds with $\le$ in place of 
equality, we say $x$ is \emph{lower weakly majorized} by $y$, $x \prec_w y$. 

We can extend the majorization relation to the set of 
all ``vectors'' (i.e. $1 \times n$ row matrices) of finite length ($n$ 
not fixed) by padding the shorter vector with zeros and applying the
above definition.

\begin{theorem}\label{theorem: substochastic weak majorization}
An $n \times n$
matrix $M$ is doubly substochastic iff $y \in \R^n_+ \implies (My \in 
\R^n_+  ~\&~ My \prec_w y)$.  
\end{theorem}
This is C.3 on page 39 of \cite{MarshallOlkin}.

\begin{definition}
A function $f: \R^n \rightarrow \R$ is called 
\emph{Schur-concave} if for every $\v, \w \in \R^n$, $\v$ 
majorizes $\w$ implies $f(\v) \le f(\w)$. 
\end{definition} 

\begin{proposition} Every concave symmetric function is Schur-concave.
\end{proposition}

An effect is called \emph{atomic} if it is on an extremal ray of $A^*_+$
(the cone of normalized effects) and is the maximal normalized effect
on that ray.  It is equivalent, in the projective context, to say it
is an atom in the orthomodular lattice of projective units.  In a
projective state space the projective units are precisely the extremal
points of the convex compact set $[0,u]$ of effects. (This is the 
finite-dimensional case of 
Proposition 8.31 of \cite{ASBook}.)

In projective state spaces, elements of $A^*$ have a spectral
expansion in terms of mutually orthogonal projective units.  We may
take these units to be atomic, but then the expansion coefficients may
be degenerate.  We can also choose the expansion so that the
coefficients are nondegenerate (the projective units no longer being
necessarily atomic); the nondegenerate expansion is unique.  In
Appendix \ref{appendix: spectral expansion} these facts are shown to be the
finite-dimensional case of \cite{ASBook}, Corollary 8.65.  We leave
open whether the expansion is unique in the stronger sense analogous
to that in (Unique) Spectrality, that the probabilities in any 
expansion into atomic effects (though not necessarily the
expansion itself) are unique.

One wonders whether analogous things hold for the \emph{states} of a
projective system.  Weak Spectrality probably does hold (and
perhaps some analogue of weak spectral decomposition for arbitrary
elements of $A$ therefore follows, i.e. one without uniqueness).  But
there are clear counterexamples to the conjecture that projectivity
implies (Unique) Spectrality.  Strictly convex sets in finite
dimension are spectral convex sets in the sense of \cite{ASBook},
Ch. 8, and therefore normalized state spaces of projective systems.
But one can easily construct one even in two affine dimensions in
which there is a state with two distinct convex decompositions into
perfectly distinguishable (``antipodal'') pure states, having very
different probabilities.  A non-equilateral isosceles triangle that
has been perturbed (``puffed out'') to be strictly convex (and
therefore spectral, and the base for 
a projective system) does the trick.  One can even
construct an example (not strictly convex, but still spectral in the
sense of \cite{ASBook}, Ch. 8) in which there is a state with convex
decompositions into different numbers of perfectly distinguishable pure
states.  See Theorem 8.87 of \cite{ASBook}.  For the special case of
the family of sets constructed in that theorem, illustrated in their
Fig. 8.1, the ``triangular pillow'' (an equilateral triangle puffed
into a third dimension), the state at the barycenter of the
equilateral triangle (with vertices the three pure states in the
``equatorial plane'') can be written as the sum of $1/3$ times each of
the three vertices, or of $1/2$ times the ``north pole'' plus $1/2$
times the ``south pole''.  It would be nice to know whether or
not this state space is self-dual.

For a projective state space, every atomic effect takes the value 
$1$ on a \emph{unique} normalized state, which is extremal in $\Omega$,
called $\widehat{e}$, and every extremal normalized state takes the
value $1$ on a unique atomic effect, called $\widetilde{\omega}$.  ~ $\widehat{}$ 
~ and ~ $\widetilde{}$ ~ are 1-1 maps of the atomic effects onto the pure states
and vice versa, and are each others' inverses.  For a pair of 
states $\omega, \sigma$, $\widetilde{\omega}(\sigma)$ is sometimes called
the \emph{transition probability} from $\sigma$ to $\omega$.  

\begin{definition}[Symmetry of Transition Probabilities]
A system is said to satisfy \emph{Symmetry of Transition Probabilities} 
(or Axiom {\bf STP}) if for any pair of pure states $\omega, \sigma$, 
$\widetilde{\omega}(\sigma) = \widetilde{\sigma}(\omega)$.
\end{definition}

We call a measurement \emph{fine-grained} if all of its effects are 
proportional to atomic effects (cf. \cite{Entropy, SWentropy}). 

\begin{theorem}\label{theorem: majorization}
Let a system satisfy Unique Spectrality, Symmetry of
Transition Probabilities, and Projectivity.  Then for
any state $\omega$ it holds that for any fine-grained measurement
$e_1,...,e_n$, the
vector ${\mathbf p} = [e_1(\omega), ..., e_n(\omega)]$ 
of probabilities of the measurement outcomes is majorized by
the vector of probabilities of outcomes for a spectral measurement
on $\omega$.
\end{theorem}

\begin{lemma}\label{lemma: distinguishability}
For a system satisfying Projectivity, $\omega$ is perfectly
distinguishable from $\sigma$ if, and only if, $\face(\omega)
\subseteq \face(\sigma)'$.
\end{lemma}

\begin{proof}[of Lemma \ref{lemma: distinguishability}] 
It follows straightforwardly from the definition of filters and their
complements, that for $P$ the filter associated with $\face{(\sigma)}$
and $P'$ the filter associated with $\face(\sigma')$, the projective
units $u \circ P$ and $u \circ P'$ distinguish $\sigma$ from $\omega$.
\qed
\end{proof}

\begin{proof}[of Theorem \ref{theorem: majorization}] 
Let $\omega = \sum_j p_j \omega_j$ be a convex decomposition, 
with $\omega_j$  pure and
perfectly distinguishable.  Then $p_j$ are the outcome probabilities
for a spectral measurement on $\omega$.  Write the effects $e_i$ of an
arbitrary fine-grained measurement as $e_i = c_i \pi_i$, where $0 <
c_i \le 1$ and $\pi_i$ are atomic.  Then the outcome probabilities for
this measurement, made on $\omega$, are 
\beq
q_{i} = \sum_j p_j c_i \pi_i (\omega_j)\;.
\eeq
That is, $q = Mp$ where $M_{ij} = c_i \pi_i (\omega_j)$.
So $\sum_{i=1}^M M_{ij} = \sum_i c_i \pi_i (\omega_j) = u(\omega_j) = 1$.
That is, $M$ is row-stochastic.  Also 
$\sum_{j=1}^N M_{ij} = c_i \pi_i (\sum_j \omega_j)$.  
By Symmetry of Transition Probabilities this is equal (using the fact
that ~$\widehat{}$~ and ~$\widetilde{}$~ are inverses) to 
$c_i \sum_j \widetilde{\omega}_j ( \widehat{\pi_i})$.  
Since $\omega_i$ are orthogonal pure states, $\widetilde{\omega}_i$ 
are orthogonal projective units \cite{ASBook}, whence 
$\sum_j \widetilde{\omega}_j \le  u$, whence 
$c_i \sum_j \widetilde{\omega}_j ( \widehat{\pi_i}) \le  c_i u(\widehat{\pi}_i)
\le c_i$.  So,  $\sum_{j=1}^N M_{ij} \le c_i$.  $c_i \le 1$, so  $M$ is 
column-substochastic.

So $M$ is doubly substochastic.  Letting $R \ge N$ be the number of
outcomes of the finegrained measurement, we pad ${p}$ with $R - N$
zeros and pad ${M}$ on the right with $R-N$ zero columns to obtain a
doubly substochastic matrix $\tilde{M}$.  Then $\tilde{M} \tilde{p} =
{q}$, so by Theorem \ref{theorem: substochastic weak majorization}
${q} \prec_w {p}$.  Since $\sum_i p_i = \sum_i q_i = 1$, lower weak
majorization implies majorization, ${q} \prec {p}$.  \qed
\end{proof}

\begin{corollary}
In a perfect system satisfying Axiom {\bf S}, 
for any state $\omega$ the outcome probabilities for any
fine-grained measurement on $\omega$ are majorized by those for a
spectral measurement on $\omega$.  In particular, this is so for
systems satisfying Postulates 1 and 2 of \cite{BMU}.
\end{corollary}

The first statement holds because, as we will show in Section
\ref{sec: projective stp perfect}, perfect systems are the same thing
as projective systems satisfying {\bf STP}.
The second sentence holds because Postulates 1 and 2 of
\cite{BMU} imply both {\bf P} and {\bf S}.  While we shall see that 
perfection implies weak spectrality, we do not know whether it implies
{\bf S}, so {\bf S} had to be included in the premise of the Corollary.

\begin{corollary}\label{cor: mixture of reversible}
Let $\omega' = \int_K d \mu(T) T(\omega)$, where $d\mu(T)$ is a normalized
measure on the compact group $K$ of reversible transformations.  In a perfect
system satisfying {\bf S}, $\omega' 
\preceq \omega$.  
\end{corollary}
\begin{proof}
Let $e_i$ be the spectral measurement on $\omega'$.  Then $e_i(
\omega') = \int_K d \mu(T) e_i(T(\omega))$.  For any state $\sigma$, write
${\bf p}$ for the vector whose $i$-th entry is $e_i(\sigma)$.  Then the
spectrum of $\omega'$ is ${\bf p}(\omega')$, and ${\bf p}(\omega') 
\equiv \int_K d \mu(T) {\bf p}(T(\omega))$ but ${\bf p}(T(\omega))$ is
just the vector of probabilities for the finegrained 
measurement $\{T^{\dagger} e_i \}$
on $\omega$, hence is majorized by the spectrum of
$\omega$.  A limit of convex combinations of such things, for
example the spectrum of $\omega'$, also majorizes the spectrum of
$\omega$. \qed
\end{proof}

\begin{definition}[\cite{Entropy, SWentropy}]
The \emph{measurement entropy} $S_{meas}(\omega)$ of a state $\omega$
of a system $A$ is defined to be the infimum, over finegrained measurements,  
 of the entropy of the outcome probabilities of the
measurement.
\end{definition}

Recall that $S(\omega)$ 
is defined, for any Axiom-{\bf S} theory, as the entropy of
the probabilities in any convex decomposition of $\omega$
into perfectly distinguishable pure states.
From the definitions of $S(\omega)$ and $S_{meas}(\omega)$, Theorem 
\ref{theorem: majorization}, and the
Schur-concavity of entropy we immediately obtain:

\begin{proposition} 
For any state $\omega \in \Omega_A$ 
of a system satisfying Axioms S, P, and STP, equivalently (see the
next two sections) satisfying 
Axiom S and perfection, $S(\omega)= S_{meas}(\omega)$.
\end{proposition} 

From this we immediately obtain, by Theorem 1 in \cite{Entropy}
(concavity of the measurement entropy; there is probably a similar
theorem in \cite{SWentropy}), that $S(\omega)$ is concave in any
system satisfying Axioms {\bf S}, {\bf P}, and {\bf STP} (and hence in
any system satisfying Postulates 1 and 2 (WS and Strong Symmetry) of
\cite{BMU}).  Similarly, under these assumptions any entropy-like
quantity constructed by applying a Schur-concave function $\chi$ to
the spectrum will be the same as the infimum of $\chi$ over
probabilities of measurement outcomes, and if $\chi$ is concave, so
will be the function $\omega \mapsto \chi(\lambdad(\omega))$.

\section{Spectral expansions of observables in projective systems}
\label{spectral expansions}

In proving Theorem \ref{theorem: projective stp is perfection} in
Section \ref{sec: projective stp perfect} (that any
(finite-dimensional) projective system satisfying symmetry of
transition probabilities is perfect) we will use the following fact,
which is of interest in its own right.  It is a dual-space
analogue of weak spectrality, but not obviously equivalent to it.

\begin{proposition}\label{prop: spectral expansion of observables}
In a projective system each $a \in A^*$ is a linear combination
$\sum_{i=1}^n \lambda_i p_i$ of mutually orthogonal projective units $p_i$.  
We can always choose the expansion
so that the coefficients $\lambda_i$ are nondegenerate (i.e. $i \ne j \implies
\lambda_i \ne \lambda_j$), and then the
expansion is unique.
\end{proposition}

We call the unique expression for $a$ as a nondegenerate linear
combination of mutually orthogonal projectors its \emph{spectral
  expansion}.  As shown in Appendix \ref{appendix: spectral
  expansion}, Proposition \ref{prop: spectral expansion of observables}
is the finite-dimensional case of the following, plus 
an observation concerning uniqueness 
following from uniqueness of the family of 
projective units in \cite{ASBook}, Theorem 8.64):

\begin{proposition}[\cite{ASBook}, Corollary 8.65]
If $A$ and $V$ are in spectral duality, then each $a \in A$ can be
approximated in norm by linear combinations $\sum_{i=1}^n \lambda_i
p_i$ of mutually orthogonal projective units $p_i$ in the
$\cP$-bicommutant of $a$.
\end{proposition}

As noted in the discussion following the definition of spectral
duality, Definition 8.42 of \cite{ASBook}, by Theorem 8.72 of
\cite{ASBook}, in finite dimension their property of spectral duality
is equivalent to all exposed faces of $\Omega$ being projective,
i.e. in our terminology, to the system $A$ being projective.  

\section{For projective systems, symmetry of transition 
probabilities is perfection} 
\label{sec: projective stp perfect}

The following theorem shows that we may replace the conjunction of
Projectivity and Symmetry of Transition Probabilities with the the
property of perfection.  Parts of our proof are modeled after the
proof of Lemma 9.23 of \cite{ASBook}, but with different assumptions:
finite dimension makes certain things simpler for us, but our premise
involves only projectivity and symmetry of transition probabilities,
not the additional property of purity-preservation by compressions
that figures in said Lemma.  After proving the theorem, we realized
that essentially the same result, stated in somewhat different terms,
was proved by H. Araki in \cite{Araki80}.  We include our proof in
Appendix \ref{appendix: projective stp is perfection}.

\begin{theorem}\label{theorem: projective stp is perfection}
Let $A$ be a finite-dimensional projective system satisfying
Symmetry of Transition Probabilities.  Then there is a unique positive
linear map $\phi: A^* \rightarrow A$ such that $\phi(x) = \hat{x}$ for
each atomic projective unit $x$.  $(x,y)$, defined by $(x , y) := \langle x,
\phi(y) \rangle$, is an inner product on $A^*$, with respect to which
compressions are symmetric, i.e.: \beq (Pa , b) = (a , Pb).  \eeq
Hence $A^*_+$ (and so also $A_+$) is a perfect self-dual cone, so the
system $A$ is a perfect system.
\end{theorem}

\begin{corollary}
For projective systems, symmetry of transition probabilities is
equivalent to perfection.
\end{corollary}
\begin{proof}
Theorem \ref{theorem: projective stp is perfection} 
gives one direction; the other direction,
that perfect projective systems satisfy STP,  is near-trivial (cf.
\cite{BMU}).  
\end{proof} 

\begin{corollary}
For a projective system satisfying STP any element $x \in A$ has a
``finegrained spectral expansion'' $x = \sum_i \lambda_i \omega_i $, with
$\lambda_i \in \R$ and $\omega_i$ mutually orthogonal pure states in
$\Omega$.
\end{corollary}

This follows from the spectral expansion of observables (Proposition 
\ref{prop: spectral expansion of observables}) and Theorem 
\ref{theorem: projective stp is perfection}
since the latter implies self-duality.   
The uniqueness properties of the expansion of 
elements of $A^*$ (cf. discussion following Proposition 
\ref{prop: spectral expansion of observables}) also hold for elements
of $A$ since the expansion in the state space will be the image 
of the expansion of Proposition \ref{prop: spectral expansion of observables}
under the order-isomorphism 
$\phi: A^* \rightarrow A$.

\begin{question} 
Does perfection imply the stronger uniqueness
properties embodied in Axiom {\bf S}? 
\end{question}
If the triangular pillow based state space (see Sec. 
\ref{sec: majorization})  is perfect, then it does
not.

\section{Filters, compressions, and von Neumann's argument for entropy}
\label{SecvonNeumann}

We have given assumptions that imply the existence of a spectral entropy
and related quantities with operational interpretation in terms 
of probabilities of measurement entropy, and majorization
properties, such as Corollary \ref{cor: mixture of reversible}, that in the 
quantum case play a crucial role in thermodynamic resource theory.  
We would like to use the spectrum and associated entropic quantities
and majorization relations in a generalized thermodynamic resource theory.
In this section we take a step in this direction by extending 
von Neumann's argument that $S(\rho)$ is the thermodynamic entropy in quantum 
theory, to systems whose internal state space is a more general GPT state space
satisfying Axiom {\bf S} and Axiom {\bf P}.

Von Neumann's argument is that a reversible process, making
use of a heat bath at temperature $T$, exists between a system with
density matrix $\rho$, and a system in a definite pure state, and that
this process overall involves doing work $-kT \tr \rho \ln{\rho}$ in
the forward direction.  His argument involves a system with both
quantum and classical degrees of freedom, e.g. a one-molecule ideal
gas, and the direct heat exchange and doing of work involves classical
degrees of freedom (specifically, expanding or contracting the volume
occupied by a gas, while in contact with the heat reservoir).

Consider an arbitrary state $\omega = \sum_i q_i \omega_i$, where
$q_i, \omega_i$ are a convex expansion of the state into a set of
$N$ perfectly distinguishable pure states $\omega_i$.  Axiom {\bf S}
ensures that such expansions exist and that they uniquely define (as
in the preceding section, except with different units corresponding to
taking $\ln$ instead of the base-2 $\log$) $S(\omega) := - \sum_i q_i
\ln q_i$.  By Lemma \ref{lemma: distinguishability} 
Axiom {\bf P} ensures that there
are atomic projective units $\pi_i$ that form a measurement
distinguishing the states $\omega_i$ (that is to say, $\sum_i \pi_i =
u$, and $\pi_i(\omega_j) = \delta_{ij}$).  There are associated
filters $P_i$ such that $P_i \omega_j = \delta_{ij} \omega_j$.

{\bf Assumption:} \emph{if such filters exist, they allow us, at no 
thermodynamic cost, to take a
  container of volume $V$, containing such a particle in equilibrium at
  temperature $T$ and separate it into $N$ containers $C_i$ of volume
  $V$, still at temperature $T$, such that $C_i$ contains a system in
  state $\omega_i$, with probability $p_i$, in which case the other
  containers are empty.}

We may think of this separation process as realizing the measurement
by using the instrument $\{ P_i \}$, with the classical measurement
outcome recorded as which container the system is in.  Because
of projectivity, it is consistent
to assume that this is possible, and that the state of the system in
container $i$ (i.e., conditional on measurement outcome $i$) is still
$\omega_i$ due to the neutrality property of filters.  
Von Neumann's argument involves instead semipermeable membranes, allowing
particles whose internal state is $\proj{i}$ to pass (from either direction), 
whilst reflecting particles whose internal state is supported on the subspace
of Hilbert space orthogonal to $\proj{i}$.  The use of analogous 
semipermeable membranes,
in the GPT case, which behave differently for systems in face $F_i$ 
(i.e., whose state ``is'' $\omega_i$) than for 
particles whose internal state ``is'' in $F_i'$, will allow us to
ultimately separate each of the mutually distinguishable states
$\omega_i$ into its own container.  We may, if we like,
represent such a procedure  
by a transformation on a tensor product of a 
classical state space and the internal state space of the particle, 
for example: 
\beq
x \otimes \omega \mapsto (x \oplus 0) \otimes P_i \omega + 
(x \oplus 1) \otimes P_i' \omega, 
\eeq
easily verified to be positive and base-norm-preserving.

In fact, the overall process that separates particles into the
containers $C_i$ could just be represented as the positive map: 
\beq 
T: x \otimes \omega \mapsto \sum_i (x + i-1)_{\mathrm{mod~}{N}} \otimes P_i \omega 
\eeq
where the first register is classical and takes values $1,...,N$,
where $N$ is the maximal number of distinguishable states.  Again this
is positive and trace preserving on the overall tensor product of the classical
$N$-state system with the GPT system (which is equivalent in structure
to the direct sum of $N$ copies of the GPT system).  The possibility
of such transformations is due to the projectivity of the state space
(which implies such properties as the neutrality of filters).  Whether
or not it is reasonable to consider them thermodynamically costless is
less clear, especially because the overall transformation on the
GPT-classical composite is not in general an automorphism of the
normalized state space (not ``reversible'').  Ultimately, the
reasonableness of this assumption probably requires that the
``measurement record'' be kept in a system for which the overall
measurement dynamics on the composite with the original system, can be
reversible, a property which we are investigating.  (This is related
to the notion of purification used in \cite{Giulio1}, 
\cite{Giulio2}.) 

Obviously, if $\omega = \sum_i q_i \omega_i$ where $\omega_i \in F_i$,
then we have $T(1 \otimes \omega) = \sum_i q_i i \otimes \omega_i$.

At this point (or after the next step, it does not matter), we
``adiabatically'' transform the internal state $\omega_i$ of the
particle in each container $C_i$, to some fixed $i$-independent pure
state, $\omega_0$.  

{\bf Assumption} (``Adiabatic assumption''): This can be achieved 
without doing any work on the system, or exchanging any heat with the 
bath.  (Thus we could do it while the system is isolated.) 

If the reversible transformation group of the GPT system (the group of
permitted transformations that are in the automorphism group of
$\Omega$) acts transitively on the pure states, that would motivate
this assumption.  This would follow, for example, from the much
stronger property of Strong Symmetry from \cite{BMU}.

Next, we isothermally (in contact with the heat bath at temperature
$T$) compress the contents of each container $C_i$ to a volume $V_i :=
q_i V$.  The work done on container $i$, \emph{if} it contains a
particle, is $W_i = - \int_{V}^{qV} P dV$.  By the ideal gas law,
$PV=nkT$; with $n=1$, $P = kT/V$ so $W_i = - kT \int_{V}^{qV} (1/V)dV
= - kT (\ln qV - \ln V) = - kT \ln q \ln V / \ln V = - kT \ln q$.
Since the probability the container contains the particle is $q_i$, and
these events, for the various containers $i$, are mutually exclusive
and exhaustive, the expected work done in this step is $\sum_i q_i W_i
= - kT \sum_i q_i \ln q_i = S(\omega)$.

Then we put the containers of compressed volume $q_i V$ next to each
other and remove partitions separating them, obtaining a particle
whose internal GPT degree of freedom is in the pure state $\omega_0$,
in equilbrium at temperature $T$ and in the volume $\sum_i q_i V = V$.
Since this process was reversible, we see that we may go from a
particle in volume $V$ and state $\omega$, whose spatial degrees of
freedom are in equilibrium at temperature $T$, to one in spatial 
equilibrium at temperature $T$ and volume $V$ but with state $\sigma$,
by doing expected work $S(\omega) - S(\sigma)$ on the particle. \qed


Just because our assumptions imply we can do this process, does not
imply that we have a consistent thermodynamics (i.e. one without
work-extracting cycles).  It is possible that further properties of a
GPT beyond projectivity and spectrality might be necessary for this.
A notion of \emph{energy}, such as Postulate 4 (Energy Observability)
in \cite{BMU} provides, would be needed for a thermodynamics that
resembles our current thermodynamic theories, if we wish to discuss
work done by or on GPT systems.  We already mentioned the principle
that if there exists a reversible transformation (automorphism of the
normalized state space $\Omega$), it can be applied
with the system isolated or in contact with a heat bath, at no cost in
work and with no heat exchange, and used along with a further
property, that the automorphism group of $\Omega$ acts transitively on
pure states, to motivate assuming zero thermodynamic cost for a step
in the von Neumann protocol.  Perhaps we can find a similar motivation
for the Strong Symmetry axiom of \cite{BMU}.  In \cite{BMU} it was
shown that the absence of higher-order interference was equivalent,
given Weak Spectrality and Strong Symmetry, to the postulate that
filters take pure states to multiples of pure states.  This
purity-preservation property greatly constrained state spaces, giving
(when conjoined with {\bf WS} and {\bf SS}) irreducible or classical Jordan systems, and at first blush it seems
it might be important for thermodynamics.  We suspect that it is not,
and that a robust thermodynamics may be developed for GPT systems that
share many of the remarkable properties of quantum theory, but are
distinctly non-quantum in their interference behavior.



\nocite{*}
\bibliographystyle{eptcs}
\bibliography{generic}

\begin{thebibliography}{1}

\bibitem{SecondLaws} F. G.S.L. Brandao, M. Horodecki, N. H. Y. Ng, 
J. Oppenheim, S. Wehner, \emph{The second laws of quantum thermodynamics},
Proc. Natl. Acad. Sci. {\bf 112}, 3275 (2015). 
\doi{10.1073/pnas.1411728112}
{\tt arXiv:1305.5278.}


\bibitem{LJR} M. Lostaglio, D. Jennings and T. Rudolph, 
\emph{Description of coherence in thermodynamic processes requires constraints beyond free energy}, Nature Communications {\bf 6}, 6383 (2015).
\doi{10.1038/ncomms7383}
{\tt arXiv:1405.2188.}

\bibitem{InPreparation} H.\ Barnum, J.\ Barrett, M.\ Krumm, and M.\ P.\ M\"uller, in preparation.

\bibitem{BMU} H. Barnum, M. M{\"u}ller and C. Ududec, \emph{Higher-order
interference and single-system postulates characterizing quantum theory},
New J.\ Phys.\ \textbf{16}, 123029 (2014). 
\doi{10.1088/1367-2630/16/12/123029}
{\tt arXiv:1403.4147.}


\bibitem{Krumm} M.\ Krumm, \emph{Thermodynamics and the Structure of Quantum Theory as a Generalized Probabilistic Theory},
Master Thesis, Heidelberg University, April 2015.

\bibitem{Giulio1}
G.\ Chiribella and C.\ M.\ Scandolo, \emph{Entanglement and thermodynamics in general probabilistic theories}, {\tt arXiv:1504.07045.}

\bibitem{Giulio2}
G.\ Chiribella and C.\ M.\ Scandolo, \emph{Operational axioms for state diagonalization}, {\tt arXiv:1506.00380.}

\bibitem{Wilce10}
A. Wilce, \emph{Formalism and interpretation in quantum theory}.  
Found.\ Phys.\ {\bf 40}, 434-462 (2010). 
\doi{10.1007/s10701-010-9410-x}

\bibitem{Wilce09}
A. Wilce, \emph{Four and a half axioms for quantum mechanics}.
In Y. Ben-Menahem \& M. Hemmo, eds. \emph{Probability in Physics}.
Springer (2012).  
\doi{10.1007/978-3-642-21329-8}
{\tt arxiv.org/abs/0912.5530}.

\bibitem{MDV}
M.\ P.\ M\"uller, O.\ C.\ O.\ Dahlsten, and V.\ Vedral, \emph{Unifying typical entanglement and coin tossing:
on randomization in probabilistic theories}, Commun.\ Math.\ Phys.\ \textbf{316}(2), 441--487 (2012). 
\doi{10.1007/s00220-012-1605-x}
{\tt arXiv:1107.6029.}

\bibitem{MOD}
M.\ P.\ M\"uller, J.\ Oppenheim, and O.\ C.\ O.\ Dahlsten, \emph{The black hole information problem beyond quantum theory},
JHEP \textbf{09}, 116 (2012). 
\doi{10.1007/JHEP09(2012)116}
{\tt arXiv:1206.5030.}

\bibitem{MarshallOlkin} A.\ W.\ Olkin, I.\ Marshall and B.\ Arnold,
  \emph{Inequalities: Theory of Majorization and its Applications},
  2nd edition, Springer (2011).
\doi{10.1007/978-0-387-68276-1\_21}

\bibitem{Entropy} H.\ Barnum, J.\ Barrett, L.\ Clark, M.\ Leifer,
  R.\ W.\ Spekkens, N.\ Stepanik, A.\ Wilce, and R.\ Wilke,
  \emph{Entropy and information causality in general probabilistic
    theories}, New J.\ Phys.\ {\bf 12}, 033024 (2010). 
\doi{10.1088/1367-2630/12/3/033024}
{\tt arXiv:0909.5075.}  Addendum, New J.\ Phys.\ {\bf 14} 129401 (2012).
\doi{10.1088/1367-2630/14/12/129401}

\bibitem{SWentropy} A.\ J.\ Short and S.\ Wehner, \emph{Entropy in
  general physical theories}, New J.\ Phys.\ {\bf 12}, 033023 (2010). 
\doi{10.1088/1367-2630/12/3/033023}
{\tt arXiv:0909.4801.}

\bibitem{ASBook} E.\ M.\ Alfsen and F.\ W.\ Shultz, \emph{Geometry of State Spaces of Operator Algebras}, Birkh\"auser, Boston, 2003.
\doi{10.1007/978-0-387-68276-1}


\bibitem{IochumThesis} B.\ Iochum, \emph{C$\hat{o}$nes Autopolaires dans
les espaces de Hilbert}, Th{\`{e}}se de 3\textsuperscript{{\`{e}}me} cycle, Marseille, 1975.

\bibitem{IochumBook} B.\ Iochum, \emph{C$\hat{o}$nes Autopolaires et Alg\`ebres de Jordan}, Springer (1984).
\doi{10.1007/BFb0071358}

\bibitem{Araki80} H. Araki, \emph{On a characterization of the state space of quantum mechanics}, Comm. Math. Phys. {\bf 75}, 1-24 (1980). 
\doi{10.1007/BF01962588}


\end{thebibliography}

\begin{appendix}

\section{Proof of Proposition \ref{prop: spectral expansion of observables}}
\label{appendix: spectral expansion}

In order to prove Proposition 
\ref{prop: spectral expansion of observables} we state 
(with very minor notational changes, and the
incorporation of some definitions that Alfsen and Shultz place in
surrounding text) Alfsen and Shultz' theorem of which it is a 
corollary (note that it is $V$ that corresponds to what we've been
calling $A$; $A$ corresponds to what we call $A^*$):

\begin{theorem}[\cite{ASBook}, Theorem 8.64]\label{theorem: AS spectral}
Assume $A$ and $V$ are in spectral duality, and let $a \in A$.  Then there
is a unique family $\{ e_\lambda \}_{\lambda \in \R}$ of projective units with 
associated compressions $P_\lambda$ such that 
\begin{enumerate}[(i)]
\item $e_\lambda$ is compatible with $a$ for each $\lambda \in \R$,
\item $P_\lambda a \le \lambda e_\lambda$ and $P'_\lambda \ge \lambda e'_{\lambda}$
for each $\lambda \in \R$,
\item $e_\lambda = 0$ for $\lambda < - ||a||$, and $e_\lambda = 1$ for 
$\lambda > ||a||$,
\item $e_\lambda \le e_\mu$ for $\lambda < \mu$,
\item $\bigvee_{\mu > \lambda} e_\mu = e_\lambda$ for each $\lambda \in \R$.
\end{enumerate}
The family $\{e_\lambda\}$ is given by $e_\lambda = r((a - \lambda
u)^+)$, each $e_\lambda$ is in the $\cP$-bicommutant of $a$, and the
Riemann sums 
\beq s_\gamma = \sum_{i=1}^n \lambda_i (e_{\lambda_i} -
e_{\lambda_{i-1}}) \eeq 
converge in norm to $a$ when $||\gamma||
\rightarrow 0$.  Here the $\gamma$ are finite increasing sequences
$\lambda_0, \lambda_1, ... , \lambda_n$ of elements of $\R$,
satisfying $\lambda_0 < -||a||$ and $\lambda_n > ||a||$ and the norm
$||\gamma||$ of such a sequence $\gamma$ is taken to be $\max_{i \in
  \{1,...,n\}}{\lambda_i - \lambda_{i-1}}$ (Note that $n$ depends on
$\gamma$, indeed $n \rightarrow \infty$ is necessary to secure $||\gamma|| \rightarrow 0$ given
the bounds on $\lambda_0$ and $\lambda_n$.)
\end{theorem}

\noindent
{\bf Proof of Proposition \ref{prop: spectral expansion of observables}:} 
As already mentioned, in finite dimension spectral
duality just means that the system is projective.  
Denoting by $F_p$ is the face of the cone $A$ generated by $p$, 
$q < p$ for projective units
$q$ and $p$ implies that $\lin F_q$ is a proper subspace of $\lin F_p$.  
Hence any
chain $0 =e_0 < e_1 < e_2 < e_3 < \cdots e_n= 1$ of projective units has finite
length no greater than one plus the dimension $d$ of $A$.  Thus, the family
$\{e_\lambda\}$ contains only a finite number $n \le d+1$ of distinct
projective units, which we index in increasing order as $e_0=0 < e_1 <
\cdots e_n =1$.  (Whenever we write $e$ with a roman index, the index
indexes this set; the expression does \emph{not} (except perhaps
accidentally) refer to $e_\lambda$ for the real number $\lambda = i$.)

Consider the sets $S_i := \{\lambda: \forall \mu \ge \lambda, e_{\mu}
\ge e_i \}$.  Each of these is an up-set in the ordering $\le$ of
$\R$, so it is either a closed or open half-line $[\mu_i , \infty)$ or
  $(\mu_i , \infty)$ unless $i=0$ in which case it is $\R$.  It is in
  fact closed: if it were open, $S_i = (\mu_i, \infty) \equiv
  \{\lambda: \lambda > \mu_i\}$, then by (v) of Theorem \ref{theorem:
    AS spectral}, $e_{\mu_i} = \bigvee_{\lambda \in S_i} e_\lambda$,
  so by the definition of $S_i$, $\mu_i \in S_i$, contradicting 
$S_i = (\mu_i , \infty)$.

Consequently the function $\R \rightarrow \{e_0,...,e_n\}$ which 
maps $\mu$ to $e_\mu$ is a sort of step-function; there are $n$ distinct 
 real numbers 
$\mu_1,...,\mu_n$ such that the preimage of $e_0$ is $(-\infty, \mu_1)$, 
the preimage of $e_j$ for $1<j<n$ is $[\mu_j , \mu_{j+1})$, and the 
preimage of $e_n$ is $[ \mu_n , \infty)$.  Let $\theta$ be the length 
of the shortest of these intervals; we have $0 < \theta < 2||a||/n$
[the only important thing about this seems to be that $\theta > 0$].  

Since $\gamma$ is an increasing sequence, by (iv) we have 
$e_{\lambda_i} - e_{\lambda_{i-1}} \ge 0$.  There are $n$ such
differences in the Riemann sum; at most $n \le d$ of them are nonzero;
we call them $p_1,...,p_n$.
Since $\sum_i p_i = 1$, by Proposition 8.8 of \cite{ASBook}, 
the nonzero ones are mutually orthogonal.    

All sequences $\gamma$ with $||\gamma|| < \theta$ have the
same finite set of nonzero differences 
$p_i := e_{\lambda_i} - e_{\lambda_{i-1}}$, 
of cardinality $n \le d$.   For such $\gamma$, 
the Riemann sums $s_\gamma$ lie in the finite-dimensional
subspace of $A$ spanned by the $p_i$.  Like all subspaces of a finite-dimensional 
vector space, it is closed.  
Hence $\lim_{||\gamma|| \rightarrow 0}$ lies in this 
subspace, so it, too, is a finite linear combination of 
mutually orthogonal projective units.  
Since the family $\{e_\lambda\}_{\lambda \in \R}$ was unique, so is this linear combination.       
\qed

\section{Proof of Theorem \ref{theorem: projective stp is perfection}}
\label{appendix: projective stp is perfection}

There exists a basis $\{w_i \}$ for $A^*$ consisting of atoms; there
is a unique linear map  $\phi: A^* \rightarrow A$ that agrees with 
the map $x \mapsto \hat{x}$ on this basis.  We need to show that 
this agrees with $x \mapsto \hat{x}$ more generally, or what is the 
same thing, that it is independent of the choice of atomic basis; and
also that it is positive.    

Let $x_1,...,x_n \in A^*$ be atoms, and $\lambda_1,...,\lambda_n \in \R$.  
By STP, for each atom $y$
\beq \label{eq: an equation}
\langle y , \sum_{i=1}^n \lambda_i \hat{x}_i \rangle = 
\langle \sum_{i=1}^n \lambda_i x_i , \hat{y} \rangle .
\eeq
So if 
$\sum_i \lambda_i x_i = 0$, then $\sum_i \lambda_i \hat{x}_i = 0$, where
$x_i$ are any atoms.  Now let $x$ be an arbitrary atom, and expand it
in the basis of atoms $w_i$:  $x = \sum_i \alpha_i w_i$.  So 
$x - \sum_i \alpha_i w_i = 0$, so $\hat{x} = \sum_i \alpha_i \hat{w}_i
\equiv \sum_i \alpha_i \phi(w_i) = \phi(x)$.  Since $\phi$ 
takes the 
set of all atomic effects (which generate the cone $A^*_+$) 
to the set of all extreme points of $\Omega$ (which generate the cone 
$A_+$), it takes $A^*_+$ onto $A_+$, so it is an order-isomorphism.
By $\phi$'s linearity the form $(. ,. ) := \langle . , \phi(.) \rangle$ 
is bilinear.  Since an order-isomorphism is in particular 
an isomorphism of linear spaces, the form $(.,.)$ is nondegenerate.    
That 
it is symmetric is
easy to see from STP:  for arbitrary $a = \sum_i a_i w_i$, 
$b = \sum_j b_j w_j$, 
\beq
(a,b) = \sum_{ij} a_i b_j \langle w_i , \phi(w_j) \rangle
\eeq
but since $w_i, w_j$ are atoms, $\phi(w_j) = \hat{w}_j$, and by 
STP $\langle w_i , \hat{w}_j \rangle = \langle w_j , \hat{w}_i \rangle
= \langle w_j , \phi(w_i) \rangle$, 
we have
\beq
(a,b) = \sum_{ij} a_i b_j \langle w_j , \phi(w_i) \rangle
= \langle \sum_j b_j w_j , \sum_i a_i \phi(w_i) \rangle = (b,a)\;.
\eeq

To establish that $(. ,. )$ is an inner product, it remains to be shown 
that $(x, x) \ge 0$ for all $x \in A^*$.  To see this, use the spectral 
expansion $x = \sum_i \lambda_i p_i$, $\lambda_i \in \R$, $p_i$ mutually
orthogonal atoms, afforded by Proposition 
\ref{prop: spectral expansion of observables}.  Then 
\beqa
(x,x) = \sum_{ij} \lambda_{i} \lambda_j \langle p_i , \hat{p_j} \rangle
= \sum_{ij} \lambda_i \lambda_j \delta_{ij} = \sum_i \lambda_i^2 \ge 0. 
\eeqa
 
Since $\phi$ is an order-isomorphism between $A^*$ and $A$, and we have
just established that the corresponding bilinear form is an inner product, 
we have shown that $A^*$ (equivalently, $A$) is self-dual.  
   
To show that any compression $P$ is symmetric 
with respect to the form, we first establish that
\beq\label{eq: intermediate thing}
( (I - P)x , Py) = 0
\eeq
where $x,y$ are atoms.
Write the spectral expansion of $Py$, 
 $Py = \sum_i \lambda_i y_i$, with $y_i$ mutually orthogonal
atoms in $\im_+ P$ and $\lambda_i \in \R$.  Note that 
$w_i \in \im P \implies \hat{w_i} \in \im P^*$. 
(To see this (which is Eq. 9.10 in \cite{ASBook}), note that by 
the facts that compressions are normalized and positive, and the dual of 
a positive map is positive on the dual cone, for any
atom $w$ in $\im_+ P$, $P^*w = \lambda \omega$ for some $\lambda 
\in [0,1]$.  So $1 = \langle w, \hat{w} \rangle = 
\langle Pw , w \rangle = \langle w, P^* w \rangle = \lambda 
\langle w , \omega \rangle \le 1$, whence $\lambda = 1$ and
$\omega = \hat{y}$.) So 
\beqa
((I - P)x , Py) = \sum_i \lambda_i ((I - P)x , y_i)  \\
= \sum_i \lambda_i \langle (I-P) x , \hat{y_i} \rangle 
=  \sum_i \lambda_i \langle x, (I - P)^* \hat{y_i} \rangle =0.  
\eeqa
The last equality is because $\hat{y_i} \in \im P^*$.  

Now $(x , Py) \equiv ((I-P + P) x , P y) 
= ((I - P) x, Py) + (Px , Py)$, but by 
(\ref{eq: intermediate thing}), this is equal to $(Px , Py)$. 
Interchanging $x$ and $y$, we have $(y, Px) = (Py, Px)$.  But 
using symmetry of $(. ,. )$ twice: $(Px, y) = (y, Px)$ and
$(Py, Px) = (Px, Py)$ we obtain, as claimed,
$(x, Py) = (Px , y)$.

Since a symmetric projection in a real inner product space is an
orthogonal projection, we see that the
compressions on $A^*$, equivalently (when $A$ is identified with $A^*$
via $\phi$) the filters on $A$, are orthogonal projections with
respect to the self-dualizing inner product $( . , . )$.  So our cone
is perfect by a result of Iochum \cite{IochumThesis, IochumBook}: that
a self-dual cone is perfect if and only if the orthogonal (in the
self-dualizing inner product) projection $P_F$ onto the the linear
span of $F$ is positive for each face $F$.\footnote{A proof may be
  found in Appendix A of \cite{BMU}. }  \qed

\end{appendix}

\end{document}